\documentclass[runningheads]{llncs}
\usepackage{booktabs} 

\usepackage{amsmath}

\usepackage{amsthm}
\usepackage{amsfonts}
\usepackage{support_caption}
\usepackage{subcaption}
\usepackage{bm}
\usepackage{ifthen}
\usepackage{xspace}
\usepackage{dialogue}

\usepackage{silence}
\usepackage{local-macro}

\begin{document}
\title{On the benefits of being constrained when receiving signals}
%
%
\author{Shih-Tang Su\inst{1}\thanks{Supported in part by NSF grant ECCS 2038416 and MCubed 3.0.} \and
David Kempe\inst{2} \and
Vijay G.~Subramanian\inst{1}\thanks{Supported in part by NSF grants ECCS 2038416, CCF 2008130, and CNS 1955777.}}
\authorrunning{S. Su \textit{et al.}}
%
\institute{Electrical Engineering and Computer Science, University of Michigan \and
Department of Computer Science, University of Southern California}
\maketitle              
\begin{abstract}
We study a Bayesian persuasion setting in which the receiver is trying to match the (binary) state of the world. The sender's utility is partially aligned with the receiver's, in that conditioned on the receiver's action, the sender derives higher utility when the state of the world matches the action.

Our focus is on whether in such a setting, being constrained helps a receiver.
Intuitively, if the receiver can only take the sender's preferred action with smaller probability, the sender might have to reveal more information, so that the receiver can take the action more specifically when the sender prefers it.
We show that with a binary state of the world, this intuition indeed carries through: under very mild non-degeneracy conditions, a more constrained receiver will always obtain (weakly) higher utility than a less constrained one.
Unfortunately, without additional assumptions, the result does not hold when there are more than two states in the world, which we show with an explicit example.

\end{abstract}
\section{Introduction} \label{sec:intro}
In this paper, we study situations akin to the following stylized dialog, which will likely be familiar to anyone who has ever served on hiring committees:

\vspace{-6pt}
\begin{dialogue}
  \speak{Alice} I see that you wrote strong recommendation letters for your Ph.D.~graduates Carol and Dan. Can you compare them for us?
  \speak{Bob} They are both great! Carol made groundbreaking contributions to $\ldots$; Dan made groundbreaking contributions to $\ldots$.
  \speak{Alice} Which of the two would you say is stronger?
  \speak{Bob} They are hard to compare. You really need to interview both of them!
  \speak{Alice} We can only invite one of them for an interview.
  \speak{Bob} I guess Carol is a bit stronger.
\end{dialogue}
\vspace{-4pt}

What happened in this example? Alice and Bob were involved in a signaling setting, in which Bob had an informational advantage. Bob's goal was to get as many of his students interviews as possible, while 
Alice wanted to only invite the strong students. While Bob knew which of his students were strong (or how strong), Alice had to rely on 
information from Bob. As is standard in signaling settings, Bob could use this fact to improve his own utility. In this sense, the example initially was virtually identical to the standard ``judge/prosecutor'' example in the seminal paper of Kamenica and Gentzkow~\cite{kamenica:gentzkow:bayesian}.

However, a change happened along the way. When Alice revealed that she was constrained in her actions (to one interview at most), this changed the utility that Bob could obtain from his previous strategy. For example, if he had insisted on not ranking the students, Alice might have flipped a coin. Implicitly, while Bob wanted \emph{both} of his students to obtain interviews, when forced to choose, he knew he would obtain higher utility from the stronger of his students being interviewed. In this sense, his utility function was ``partially aligned'' with Alice's; this partial alignment, coupled with Alice's constraint, resulted in Alice obtaining more information, and thus higher utility.

The main goal of the present paper is to investigate to what extent the behavior illustrated informally in the dialog above arises in a standard model of Bayesian persuasion. Specifically, if the utilities of the sender and receiver are ``partially aligned,'' will it always benefit a receiver to be more constrained in how she can choose her actions?

\vspace{-8pt}
\subsection{The Model: An Overview} \label{sec:model-overview}
\vspace{-4pt}
Our model --- described 
fully in Section~\ref{sec:prob} --- is based on the standard Bayesian persuasion model of Kamenica and Gentzkow~\cite{kamenica:gentzkow:bayesian}. For our main result, we assume that the state space is binary: $\StateSpace = \SET{\State{1}, \State{2}}$. These states could correspond to a student being bad/good in our introductory example, a defendant being innocent/guilty in the example of Kamenica and Gentzkow~\cite{kamenica:gentzkow:bayesian}, or a stock 
about to go up or down.
The sender and receiver share a common prior \PRIOR for the distribution of the state \STATE.
In addition, the sender will observe the actual state \STATE, but only after committing to a signaling scheme (also called information structure).

A signaling scheme is a (typically randomized) mapping $\SIGSCHEME: \StateSpace \to \SignalSpace$.
The receiver observes the (typically random) signal $\SIGNAL = \SigScheme{\STATE}$; based on this observation, she takes an action $\ACT \in \ActSpace$.
Here, we assume that --- like the state space --- the action space is binary, i.e., $\ActSpace = \SET{\Act{1}, \Act{2}}$.
Based on the true state of the world and the action taken by the receiver, both the sender and receiver derive utilities \Util{S}{\STATE}{\ACT}, \Util{R}{\STATE}{\ACT}.
The receiver will choose her action (upon observing \SIGNAL) to maximize her own expected utility; the sender, knowing that the receiver is rational, will commit to a signaling scheme to maximize his expected utility under rational receiver behavior.

Motivated by many practical applications, we assume that the receiver prefers to match the state of the world, in the sense that
$\Util{R}{\State{1}}{\Act{1}} \geq \Util{R}{\State{1}}{\Act{2}}$
and $\Util{R}{\State{2}}{\Act{2}} \geq \Util{R}{\State{2}}{\Act{1}}$.
For instance, in our introductory example, Alice prefers to interview strong candidates and to not interview weak ones; in the judge-prosecutor example, the judge prefers convicting exactly the guilty defendants; and an investor prefers to buy stocks that will go up and sell stocks that will go down.
Our assumption about the ``partial alignment" of the sender and receiver utilities is formalized as an \emph{action-matching} preference of the sender, stated as follows: 
$\Util{S}{\State{1}}{\Act{1}} \geq \Util{S}{\State{2}}{\Act{1}}$
and $\Util{S}{\State{2}}{\Act{2}} \geq \Util{S}{\State{1}}{\Act{2}}$.
That is, if a candidate is being interviewed, Bob prefers it to be a strong candidate over a weak one (but may still prefer a weak candidate being interviewed over a strong/weak candidate \emph{not} being interviewed); similarly, if a prosecutor sees a defendant convicted, he would prefer the defendant to be guilty (but may still prefer an innocent defendant being convicted over 
going free); similarly, an investment platform may want to entice a client to buy stock, but conditioned on the client buying stock, the platform may prefer for the stock to go up.

In addition to the assumption of partial alignment, our main addition to the standard Bayesian persuasion model is to consider constraints on the receiver's actions.
Specifically, we assume that there are (lower and upper) bounds \LB, \UB on the probability with which the receiver is allowed to take action
\footnote{This implies constraints of $1-\UB, 1-\LB$ on the probability of taking action \Act{2}. A more general model and its specialization to binary actions is discussed in Section~\ref{sec:constrained-receiver}.} \Act{1}.
Such a constraint corresponds to a department only being willing to interview at most 10\% of their applicants, a judge having a quota for how many defendants (at most) to convict, or a conference having an upper bound on its number/fraction of accepted papers.
Such a constraint creates dependencies between the receiver's actions under different received signals, and may force her to randomize between actions, contrary to the standard Bayesian persuasion setting in which the receiver may deterministically choose any utility-maximizing action conditioned on the observed signal \SIGNAL.
To see this, consider a prior under which a candidate is strong with probability \third, and the receiver obtains utility 1 from interviewing a strong candidate and -1 from interviewing a weak candidate (and 0 from not interviewing). If the sender reveals no information, the receiver would prefer to interview no candidates, but a lower-bound constraint may force her to do so, in which case she would randomize the decision to interview the smallest total number of candidates.
We write $\BESTRESPONSE: \SignalSpace \to \ActSpace$ for the receiver's (typically randomized) mapping from signals to actions.
Note that the constraint applies across all sources of randomness (the state of the world, the sender's randomization, and the receiver's randomization), so it is required that $\LB \leq \Prob[\PRIOR, \SIGSCHEME, \BESTRESPONSE]{\BestResponse{\SigScheme{\STATE}}=\Act{1}} \leq \UB$.

To avoid trivialities, we assume that $\Prob[\PRIOR]{\STATE = \State{1}} \in [\LB, \UB]$, that is, if the sender revealed the state of the world perfectly, the receiver would be allowed to match it.
We say that a receiver with constraints $(\LB['], \UB['])$ is \emph{more constrained} than one with constraints $(\LB, \UB)$ iff $\LB['] \geq \LB$ and $\UB['] \leq \UB$.
\vspace{-8pt}
\subsection{Our Results}
\vspace{-4pt}
Our main result is that when the state of the world is binary, a receiver is always (weakly) better off when more constrained. We state this result here informally, and revisit it more formally (and prove it) in Section~\ref{sec:biresult}.

\begin{theorem}[Main Theorem (stated informally)]
  \label{thm:main-theorem-informal}
  Consider a Bayesian persuasion setting in which the state and action spaces are binary, the receiver is trying to match the state of the world, and the sender is action-matching.
  Then, a more constrained receiver always obtains (weakly) higher utility than a less constrained one.
\end{theorem}

Unfortunately, this insight does not extend to more fine-grained states of the world: even for a ternary state of the world, there are examples with partially aligned sender and receiver in which a more constrained receiver is strictly worse off. We discuss such an example in depth in Section~\ref{sec:general}.
It is possible to obtain some positive results recovering versions of Theorem~\ref{thm:main-theorem-informal} by imposing additional constraints on the sender's and receiver's utility functions. However, many of these constraints are strong, and have only limited applicability to real-world settings. We discuss some of these approaches in Section~\ref{sec:discuss} --- whether there are less restrictive conditions recovering Theorem~\ref{thm:main-theorem-informal} for more states of the world is an interesting direction for future work.

\vspace{-8pt}
\subsection{Related Work} \label{sec:literature}
\vspace{-4pt}
In general, information design as an area is concerned with situations in which a better-informed sender or information designer can influence the behavior of other agents via the provision of information.
The literature generally studies problems in which the underlying game between the agents is given and fixed, but where the sender can influence the outcome by an appropriate choice of information to be disclosed.
The core difference between Bayesian persuasion \cite{kamenica:gentzkow:bayesian,bergemann:morris:robust:2013,bergemann:morris:bayes:2016,kamenica:bayesian:2019,bergemann:morris:information:2019} and other standard paradigms that study information transmission (such as cheap talk~\cite{crawford:sobel:strategic}, verifiable messages~\cite{grossman:informational:1981,milgrom:good:1981} or signaling games~\cite{spence:job-signaling}) is the commitment power of the sender.
In Bayesian persuasion models, the sender moves first and commits to a (typically randomized) mapping from states of the world to signals.
Subsequently, the sender observes the state of the world and applies the mapping.
Based on the mapping and the observed signal, the rational recipients (called receivers) choose actions.

The study of Bayesian persuasion was initiated in the seminal work of Kamenica and Gentzkow~\cite{kamenica:gentzkow:bayesian} and Rayo and Segal~\cite{rayo:segal:optimalinfo:2010}.
In their work, the sender can \emph{commit to sending any distribution of messages} before (accurately) observing the state of the world; the receiver, on the other hand, only has knowledge of the prior.
The full commitment setting allows for an equivalence to an alternate model where the sender publicly chooses the amount of information (regarding the state of the world) he will privately observe and then (strategically) decides how much of this information to share with the receiver via verifiable messages.
Follow-up work of Bergemann and Morris~\cite{bergemann:morris:robust:2013,bergemann:morris:bayes:2016} established a useful and important equivalence between the set of outcomes achievable via information design and Bayes correlated equilibria.
Since these seminal works, there has been a large body of work on Bayesian persuasion with theoretical developments as well as a multitude of applications.
To keep our discussion focused, for the broader literature, we refer the reader to survey articles~\cite{kamenica:bayesian:2019,bergemann:morris:information:2019}.

The literature closest to our work studies information design with a constrained sender: the constraints arise through a diverse set of assumptions. The work in \cite{perez:prady:complicating:2012,perez:interim:2014} shows that pooling equilibria result if the receiver either prefers lower complexity (for a certification process) or performs a validation of the sender's signal; this holds whether the signals of the sender are exogenously constrained or not.
A growing body of work considers constraints on the sender that arise either due to communication costs for signaling~\cite{gentzkow:kamenica:costly:2014,hedlund:costly:2015,carroni:ferrari:pignataro:quality:2020,nguyen:tan:bayesian:2021}, capacity limitations for signaling~\cite{LimitedSignaling,letreust:tristan:persuasion:2019}, the sender's signal serving multiple purposes (such as convincing a third party to take a payoff-relevant action) \cite{boleslavsky:kim:bayesian:2018}, or costs to the receiver for acquiring additional information~\cite{matyskova:bayesian:2018}.
The contributions are then to characterize either the applicability of the concavification approach \cite{kamenica:gentzkow:bayesian}, the optimal signaling structure, or the conditions for the optimality of certain signaling structures. In \cite{kolotilin:experimental:2015}, constraints on the sender arise from the receiver having access to some publicly available information. Within this context, Kolotilin~\cite{kolotilin:experimental:2015} studies comparative statics on the sender's utility based on the quality of the sender's information or the public information.
There is also a burgeoning literature on constraints on the sender arising from a privately informed receiver (e.g., \cite{kolotilin:mylovanov:zapechelnyuk:li:persuasion:2017,guo:shmmaya:interval:2019,doval:skreta:constrained:2018,candogan:reduced:2020,candogan:strack:optimal:2021}). The main contributions in this line of research are to characterize the optimal signaling structure with a key aspect being the fact that the sender constructs a different signal for each receiver type. 

Based on the discussion above, clearly there is significant literature studying a constrained sender's optimal signaling scheme and utility. However, work that studies constraints on the \emph{receiver}, or their impact on the receiver's utility, is extremely limited. To the best of our knowledge, \cite{babichenko:cohen:zabarnyi:bayesian:2020} is the only work to analyze a constrained receiver problem. The authors impose \emph{ex ante} and \emph{ex post} constraints on the receiver's posterior beliefs, characterize the dimensionality of the optimal signaling structure and develop low-complexity approximate welfare maximizing algorithms. In our work, we have two important differences: first, we impose constraints on the receiver's \emph{actions} as opposed to posterior beliefs; and second, we explore when these constraints result in increased utility for the receiver. 
\vspace{-8pt}

\vspace{-2pt}
\section{Problem Formulation} \label{sec:prob}
\vspace{-6pt}
Our model is based on the standard Bayesian persuasion model \cite{kamenica:gentzkow:bayesian}.
Two players, a sender and a receiver, interact in a signaling game.
The sender can observe the state of the world, while the receiver can take an action. The sender can convey information about the state of the world to the sender. Both players receive utility as a function of both the state of the world and the action chosen by the receiver. Since their utility functions typically do not align, the sender will be strategic in the information he reveals to the receiver.

\vspace{-6pt}
\subsection{State of the World, Actions, and Utilities}
The (random) state of the world \STATE is drawn from a state space \StateSpace. For our main result, we assume that the state space is binary ($\StateSpace = \SET{\State{1},\State{2}}$); however, we define the model in more generality.
The sender and receiver share a common-knowledge prior distribution $\PRIOR \in \Distribution{\StateSpace}$ for \STATE. When the state space is binary, this prior is fully characterized by $\PriorProb = \Prob[\PRIOR]{\STATE = \State{1}}$.

Only the receiver can take an action $\ACT \in \ActSpace$. Again, for our main result, we assume that the action space is binary: $\ActSpace = \SET{\Act{1},\Act{2}}$.
Both the sender's and the receiver's utilities are functions of the true state \STATE and the action taken; they are captured by the functions
$\UTIL{S}: \StateSpace \times \ActSpace \to \R$ and
$\UTIL{R}: \StateSpace \times \ActSpace \to \R$.
As discussed in Section~\ref{sec:model-overview}, we assume that the receiver tries to match the state of the world with her action.%

\begin{definition}[State-Matching Receiver]
\label{def:state-matching}
We say that the receiver's utility function is \emph{state-matching} if it satisfies the following: for all $i, j, k$ with $i \leq j \leq k$ or $i \geq j \geq k$, we have that \vspace{-8pt}
\begin{align}
  \Util{R}{\State{i}}{\Act{j}} & \geq \Util{R}{\State{i}}{\Act{k}}.
\label{eqn:state-matching-receiver-general}
\end{align}                                 
When the state of the world is binary, the condition simplifies to:
\begin{align}
   \Util{R}{\State{1}}{\Act{1}} \geq \Util{R}{\State{1}}{\Act{2}}  \text{  and  } 
  \Util{R}{\State{2}}{\Act{2}}  \geq \Util{R}{\State{2}}{\Act{1}}.
\label{eqn:state-matching-receiver}
\end{align}
\end{definition}
In words, a state-matching receiver always prefers an action closer to
the true state;
however, the definition does not enforce any comparisons between choosing an action that is too high vs.~too low compared to the true state.

The key notion for our analysis is a partial alignment of the sender's utility with the receiver's. This is captured by the fact that the sender, given any fixed action, would prefer states closer to the action, expressed in Definition~\ref{def:partial-alignment}:

\begin{definition}[Action-Matching Sender]
\label{def:partial-alignment}
We say that the sender's utility function is \emph{action-matching} if it satisfies the following: for all $i, j, k$ with $i \leq j \leq k$ or $i \geq j \geq k$, we have that 
\begin{align}
  \Util{S}{\State{j}}{\Act{i}} & \geq \Util{S}{\State{k}}{\Act{i}}.
\label{eqn:partial-alignment-general}
\end{align}                                 
When the state of the world is binary, the condition simplifies to:
\begin{align}
   \Util{S}{\State{1}}{\Act{1}}  \geq \Util{S}{\State{2}}{\Act{1}} \text{ and }
  \Util{S}{\State{2}}{\Act{2}}  \geq \Util{S}{\State{1}}{\Act{2}}.
\label{eqn:partial-alignment}
\end{align}
\end{definition}

In words, an action-matching sender always prefers a state of the world closer to the action chosen by the receiver; again, we do not enforce any comparisons between states that are higher vs.~lower than the chosen action.

Notice the difference between Inequalities~\eqref{eqn:partial-alignment-general} and \eqref{eqn:partial-alignment} vs.~\eqref{eqn:state-matching-receiver-general} and \eqref{eqn:state-matching-receiver}: \eqref{eqn:state-matching-receiver-general} and \eqref{eqn:state-matching-receiver} compare the receiver's utilities when the state of the world is fixed and the action is changed, while \eqref{eqn:partial-alignment-general} and \eqref{eqn:partial-alignment} compare the sender's utilities when the action is fixed and the state of the world is changed. That is, given that the receiver takes a particular action, the sender derives higher utility when that action more closely matches the state of the world than when it does not. Again, a justification for this assumption is discussed in Section~\ref{sec:model-overview}.

\subsection{Signaling Schemes}
Before the receiver takes her action, the sender can send a signal \SIGNAL to reveal (partial) information about the state of the world.
More precisely, prior to observing the state of the world \STATE, the sender commits to a signaling scheme \SIGSCHEME, which is a mapping
$\SIGSCHEME: \StateSpace \to \Distribution{\SignalSpace}$.
For our purposes \SIGSCHEME is conveniently characterized by the probability with which each signal is sent conditional on the state.
We write $\SigProbC{i}{j} = \ProbC{\SigScheme{\STATE} = \Signal{j}}{\STATE = \State{i}} \in [0,1]$ for the probability that signal \Signal{j} is sent conditional on the state of the world being \State{i}.
We write $\SigProb{j} = \sum_i \Prob[\PRIOR]{\STATE = \State{i}} \cdot \SigProbC{i}{j}$ for the probability of sending the signal \Signal{j}.


The receiver is Bayes-rational, and her objective is to maximize her expected utility after observing the signal. The expected utility derived from action \ACT when observing \Signal{j} can be written as
\begin{align*}
\ExpUtil{R}{\Signal{j}}{\ACT}
  & =  \sum_{\State{i}\in\StateSpace}
    \ProbC{\STATE = \State{i}}{\SigScheme{\STATE} = \Signal{j}} \cdot \Util{R}{\State{i}}{\ACT}
    \;=\; \sum_{\State{i}\in\StateSpace}
  \tfrac{\Prob{\STATE = \State{i}} \cdot \SigProbC{i}{j}}{\SigProb{j}} \cdot \Util{R}{\State{i}}{\ACT}.
\end{align*}
Thus, barring other constraints (which we will introduce below), the receiver chooses an action \ACT in $\argmax_{\ACT} \ExpUtil{R}{\Signal{j}}{\ACT}$.
Following most of the literature in the field of information design, we assume that the receiver breaks ties in favor of an action most preferred by the sender. The following very useful alternative view has been observed in the prior literature (see, e.g., \cite{bergemann:morris:correlated}): instead of sending abstract signals, the sender can without loss of generality send the receiver a recommended action \Act{j}. The sender must ensure that \SIGSCHEME is such that the receiver will always voluntarily follow the recommendation. In other words, the recommended action \Act{j} must always be in $\argmax_{\ACT} \ExpUtil{R}{\Signal{j}}{\ACT}$. This constraint ensures \emph{ex-post incentive compatibility (EPIC)} of the signaling scheme, and is often referred to as an \emph{obedience constraint}.

We write $\BESTRESPONSE: \SignalSpace \to \Distribution{\ActSpace}$ for the receiver's (possibly randomized) best-response function.
In the setting described so far, there is actually no need for the receiver to randomize, and she can always choose any arbitrary deterministic $\BestResponse{\Signal{j}} \in \argmax_{\ACT} \ExpUtil{R}{\Signal{j}}{\ACT}$.
However, as we will see in Section~\ref{sec:constrained-receiver}, the situation changes when the receiver is constrained.
For a receiver strategy \BESTRESPONSE, we write $\ActProb{i}{j} = \Prob{\BestResponse{\Signal{j}} = \Act{i}}$ for the probability that the receiver, upon observing signal \Signal{j}, chooses action \Act{i}.

The sender's objective is to design a signaling strategy which maximizes his expected utility in the subgame perfect equilibrium. That is, he chooses \SIGSCHEME so as to maximize his expected utility (under all sources of randomness)
\begin{align*}
  \Expect[{\STATE \sim \PRIOR, \SIGNAL \sim \SigScheme{\STATE}, \ACT \sim \BestResponse[\SIGSCHEME]{\SIGNAL}}]{\Util{S}{\STATE}{\ACT}},
\end{align*}
assuming a best response \BESTRESPONSE[\SIGSCHEME] from the receiver.

\subsection{Constrained Receiver} \label{sec:constrained-receiver}
Our main conceptual departure from prior work is that we consider constraints on the receiver, restricting the probability with which actions can be chosen.
In a general setting, such constraints are lower and upper bounds on the probability of taking each action, i.e., \LB[\ACT] and \UB[\ACT] for each \ACT.
Formally, we require that for each action \Act{i}, the combination of the sender's signaling scheme \SIGSCHEME and the receiver's response \BESTRESPONSE satisfy \vspace{-6pt}
\begin{align}
  \LB[\Act{i}] & \leq \sum_j \SigProb{j} \cdot \ActProb{i}{j}
              \; \leq \; \UB[\Act{i}].
  \label{eqn:receiver-constraint}
\end{align}

The constraints are common knowledge among the sender and receiver.
When the state space is binary, the constraints can be simplified:
they are fully characterized by the lower and upper bounds $\LB = \max(\LB[\Act{1}], 1-\UB[\Act{2}]), \UB = \min(\UB[\Act{2}], 1-\LB[\Act{1}])$ for the probability with which the receiver can choose action \Act{1}.

The focus of our work is on whether being (more) constrained helps the receiver, by forcing an action-matching sender to disclose ``more'' information.
Without any further assumptions, this is trivially false.
For example, suppose that the state of the world is known to be \State{1} with probability 1, and both the sender and the receiver obtain utility 1 when the receiver chooses action \Act{1}, and 0 otherwise.
If the constraint specified that \Act{1} must be taken with probability 0, and \Act{2} with probability 1, then of course, the receiver (and the sender) would be worse off.
In order to allow us to clearly articulate the question of whether a constrained receiver obtains more information, we require that perfect state matching would always be feasible for the receiver, if the true state were revealed:

\begin{definition}[Implementable and Feasible Constraints]
  \label{def:feasible-constraint}
 Consider constraints \Constr{\LB[\Act{i}]}{\UB[\Act{i}]} for all $\Act{i} \in \ActSpace$.
  We say that the constraints are \emph{implementable} iff
  $\sum_i \LB[\Act{i}] \leq 1 \leq \sum_i \UB[\Act{i}]$.

  The constraints are \emph{feasible} iff $\LB[\Act{i}] \leq \Prob[\PRIOR]{\STATE = \State{i}} \leq \UB[\Act{i}]$ for all $i$.

  For the special case of a binary state space, a constraint \Constr{\LB}{\UB} is feasible iff $\LB \leq \PriorProb \leq \UB$.
\end{definition}

Notice that when constraints are not implementable, there is no strategy for the receiver to satisfy all constraints. When constraints are feasible, then with full information, perfect state matching can be implemented by the receiver.

We say that the constraints \Constr{\LB[\Act{i}]}{\UB[\Act{i}]} are \emph{more binding} (or the receiver is more constrained by them) than \Constr{\LBP[\Act{i}]}{\UBP[\Act{i}]} if and only if $\LBP[\Act{i}] \leq \LB[\Act{i}]$ and $\UB[\Act{i}] \leq \UBP[\Act{i}]$ for all $i$. 
When the state space is binary, the condition simplifies: the constraint \Constr{\LB}{\UB} is more binding than \Constr{\LB[']}{\UB[']} if and only if $\LB['] \leq \LB$ and $\UB \leq \UB[']$.

We note that the presence of a constraint may force the receiver to randomize between actions, even possibly actions that are not optimal.
For a simple example, suppose that the state of the world is binary and determined by a fair coin flip, and the receiver obtains utility 2 from matching state \State{2}, 1 from matching state \State{1}, and 0 for not matching the state. 
If the sender reveals no information, then a receiver constrained by --- say --- $\LB = \UB = \half$, would have to flip a fair coin to decide which action to choose, even though the optimal strategy would be to always choose \Act{2}.

While the receiver's best response \BESTRESPONSE may in general (have to) be randomized, we show that there is always an optimal signaling strategy for the sender such that the receiver will play a deterministic strategy \BESTRESPONSE.
Notice that the following proposition does not even require feasibility in the sense that the prior distribution satisfies the constraints: it merely requires that the constraints allow for the existence of \emph{any} signaling scheme and corresponding receiver strategy.

\begin{proposition} \label{prop:correspondingpure}
  Assume that $\SetCard{\SignalSpace} \geq \SetCard{\ActSpace}$, and let \Constr{\LB[\Act{i}]}{\UB[\Act{i}]} (for all $i$) be implementable constraints on the receiver. Then, for any signaling scheme \SIGSCHEME[\hat], there exists another signaling scheme \SIGSCHEME under which the sender has at least the same utility as under \SIGSCHEME[\hat], and such that the receiver's best response $\BESTRESPONSE[\SIGSCHEME]$ is deterministic. In particular, there is a sender-optimal strategy under which the receiver responds deterministically.
\end{proposition}

\begin{proof}
  We will give an explicit construction of such a strategy. Let \SIGSCHEME[\hat] be any signaling scheme. Let $\BESTRESPONSE[{\SIGSCHEME[\hat]}]$ be the receiver's (randomized) best response. 
  Recall that \ActProb{i}{j} is the probability with which the receiver plays \Act{i} when receiving the signal \Signal{j}. We will first construct an intermediate signaling scheme \SIGSCHEME['], and from it the final signaling scheme \SIGSCHEME.

  As a first step, the signaling scheme maps to an expanded space $\SignalSpace['] = \SignalSpace \times \ActSpace$. When observing the state \State{k}, the sender sends the signal $(\Signal{j}, \Act{i})$ with probability $ \SigProbC{k}{j} \cdot \ActProb{i}{j}$. In other words, the sender performs exactly the randomization that the receiver would perform, and makes the corresponding recommendation to the receiver. Conditioned on the signal \Signal{j}, the signal's second component \Act{i} reveals no information about the state of the world. Therefore, because the distribution of \Act{i} is exactly the distribution that \BestResponse[\hat{\SIGSCHEME}]{\Signal{j}} uses, it is a best response for the receiver (and satisfies the constraints) to deterministically\footnote{Note that it is optimal for the receiver to follow the recommendation due to the overall constraints. In isolation, the receiver may be better off deviating for some signals --- however, doing so would violate a constraint, or come at the expense of having to choose an even more suboptimal action under another signal.} follow the sender's ``recommendation'' \Act{i} when receiving the signal $(\Signal{j}, \Act{i})$.

  Then, following the standard approach for reducing the size of the signal space, we ``compress'' all signals under which the receiver chooses the same action into one signal. That is, under the final signaling scheme \SIGSCHEME, whenever the sender was going to send $(\Signal{j}, \Act{i})$ for any $j$ under \SIGSCHEME['], the sender simply sends \Act{i}. Because it is a best response for the receiver to deterministically choose \Act{i} for all received $(\Signal{j}, \Act{i})$, it is still a best response to follow the recommendation \Act{i}.

  Thus, we have constructed a signaling scheme \SIGSCHEME such that the receiver plays a deterministic best response, and the number of signals employed by the sender is at most $\SetCard{\ActSpace}$.

Finally, to prove the existence of a sender-optimal signaling scheme with deterministic receiver response, let \SIGSCHEME[\hat] be any sender-optimal signaling scheme. The existence of a signaling scheme, and thus a sender-optimal one, follows because the constraints are implementable by assumption. Then, applying the previous argument to \SIGSCHEME[\hat] gives the desired optimal signaling scheme with deterministic receiver responses.
\end{proof}

In general, most of the literature on Bayesian persuasion assumes that the signal space is at least as large as the action space (which is enough to obtain sender-optimal strategies, and find them via an LP~\cite{kolotilin:lp-optimal:2018} when EPIC holds). Hence, we make the same assumption that $\SetCard{\SignalSpace} \geq \SetCard{\ActSpace}$ in Proposition~\ref{prop:correspondingpure}. 

Henceforth, we will restrict attention to signaling schemes with deterministic best response functions \BESTRESPONSE without loss of optimality.
However, the sender still has to ensure that following the deterministic recommendation is incentive compatible for the receiver. Since the receiver is constrained, her space of deviations is only to best-response functions satisfying the constraints. This is captured by the following definition:

\begin{definition} \label{def:EAIC}
  Let $\SIGSCHEME: \StateSpace \to \SignalSpace$ be a direct signaling scheme for the sender, i.e., making action recommendations and assuming $\SignalSpace = \ActSpace$.
  Let $\Pi$ be the set of all randomized mappings $\BESTRESPONSE: \SignalSpace \to \ActSpace$ (characterized by \ActProb{i}{j}) satisfying the following inequalities for all actions \Act{j}:
  \begin{align*}
    \LB[\Act{j}] & \leq \sum_i \SigProb{i} \cdot \ActProb{i}{j}
                 \; \leq \; \UB[\Act{j}].
  \end{align*}
  Then, \SIGSCHEME is \emph{ex ante incentive compatible} iff for all feasible response functions $\BESTRESPONSE \in \Pi$,
  \begin{align*}
    \sum_i \SigProb{i} \cdot \ExpUtil{R}{\Signal{i}}{\Act{i}}
    & \geq \sum_i \sum_j \SigProb{i} \cdot \ActProb{i}{j} \cdot \ExpUtil{R}{\Signal{i}}{\Act{j}}.
  \end{align*}
\end{definition}


Note that the presence of constraints forces us to deviate from the standard EPIC requirement in the literature. Definition~\ref{def:EAIC} bears similarity to definitions in \cite{babichenko:cohen:zabarnyi:bayesian:2020,doval:skreta:constrained:2018,candogan:strack:optimal:2021}, where ex ante constraints are considered.

\vspace{-6pt}
\section{Our Main Result} \label{sec:biresult}

In this section, we present the main result of this paper.

\begin{theorem} \label{thm:main}
  Consider a Bayesian persuasion setting in which the state and action spaces are binary. The receiver is state-matching, and the sender is action-matching.
  Let \Constr{\LB}{\UB} and \Constr{\LB[']}{\UB[']} be two feasible constraints such that \Constr{\LB}{\UB} is more binding than \Constr{\LB[']}{\UB[']}, and let \SetSIGSCHEME, \SetSIGSCHEME['] be the set of all sender-optimal signaling schemes under these constraints, respectively.

Let $\SIGSCHEME \in \argmax_{\SIGSCHEME\in \SetSIGSCHEME}\Expect[{\STATE \sim \PRIOR, \SIGNAL \sim \SigScheme{\STATE}, \ACT \sim \BestResponse[\SIGSCHEME]{\SIGNAL}}]{\Util{R}{\STATE}{\ACT}}$ maximize the receiver's utility over \SetSIGSCHEME,
    and $\SIGSCHEME['] \in \argmax_{\SIGSCHEME[']\in\SetSIGSCHEME[']}\Expect[{\STATE \sim \PRIOR, \SIGNAL \sim \SigScheme[']{\STATE}, \ACT \sim \BestResponse[\SIGSCHEME']{\SIGNAL}}]{\Util{R}{\STATE}{\ACT}}$ maximize the receiver's utility over \SetSIGSCHEME['].
  Then the receiver is no worse off under \SIGSCHEME than under \SIGSCHEME['], i.e., $\Expect[{\STATE \sim \PRIOR, \SIGNAL \sim \SigScheme{\STATE}, \ACT \sim \BestResponse[\SIGSCHEME]{\SIGNAL}}]{\Util{R}{\STATE}{\ACT}} \geq \Expect[{\STATE \sim \PRIOR, \SIGNAL \sim \SigScheme[']{\STATE}, \ACT \sim \BestResponse[\SIGSCHEME']{\SIGNAL}}]{\Util{R}{\STATE}{\ACT}}.$ 
\end{theorem}
\begin{proof}
  At a high level, the intuition for the proof is as follows. Based on the discussion in Section~\ref{sec:constrained-receiver}, the constraints on the receiver actually translate into constraints on the sender in the optimization problem. Because the sender's signaling schemes are more constrained, he has to reveal more information. However, this intuition is not complete --- after all, the constraints may entice the sender to reveal \emph{less} information. Furthermore, as we see in Section~\ref{sec:general}, when the state space is not binary, a more constrained receiver may be worse off.

  Let \SIGSCHEME, \SIGSCHEME['] be 
  as defined in the statement of the theorem, and let \SigProbC{i}{j}, \SigProbC[']{i}{j} be their corresponding conditional probabilities of sending the signal \Signal{j} in state \State{i}.
By Proposition~\ref{prop:correspondingpure}, w.l.o.g., under the sender-optimal strategies \SIGSCHEME, \SIGSCHEME['], the sender recommends an action to the receiver, and the receiver deterministically follows the recommendation. That is, the signal \Signal{i} can be associated with the action \Act{i} for $i=1,2$.
Our proof is based on distinguishing four cases, depending on the sender's utility:

\begin{enumerate}
\item $\Util{S}{\State{1}}{\Act{1}} \geq \Util{S}{\State{1}}{\Act{2}}$
  and $\Util{S}{\State{2}}{\Act{2}} \geq \Util{S}{\State{2}}{\Act{1}}$

In this case, for every state, the sender prefers the same action as the receiver. Since the sender's and the receiver's preferences are fully aligned, the sender's optimal strategy is to fully reveal the state of the world. Since the constraints are feasible, the receiver can perfectly match the state of the world under both constraints, and hence obtains the same utility under both constraints.

\item $\Util{S}{\State{1}}{\Act{1}} \geq \Util{S}{\State{1}}{\Act{2}}$
  and $\Util{S}{\State{2}}{\Act{2}} \leq \Util{S}{\State{2}}{\Act{1}}$

  In this case, the sender always prefers action \Act{1}.
  Since the sender is action-matching, $\Util{S}{\State{1}}{\Act{1}} \geq \Util{S}{\State{2}}{\Act{1}}$ and $\Util{S}{\State{2}}{\Act{2}} \geq \Util{S}{\State{1}}{\Act{2}}$. Combining these inequalities, we obtain that the sender's utility function satisfies the following total order:
\begin{align*}
  \Util{S}{\State{1}}{\Act{1}}
  \geq \Util{S}{\State{2}}{\Act{1}}
  \geq \Util{S}{\State{2}}{\Act{2}}
  \geq \Util{S}{\State{1}}{\Act{2}}.
\end{align*}

This implies that
\begin{align}
  \Util{S}{\State{1}}{\Act{1}} - \Util{S}{\State{1}}{\Act{2}}
  & \geq \Util{S}{\State{2}}{\Act{1}}- \Util{S}{\State{2}}{\Act{2}}.
    \label{eqn:sender-swap}
\end{align}

We now show that $\SigProbC{1}{2} = 0$.
An identical proof also shows that $\SigProbC[']{1}{2} = 0$.
We distinguish two cases:

\begin{enumerate}
  \item If $\SigProbC{1}{2} > 0$ and $\SigProbC{2}{1} > 0$, then the sender could move some probability mass $\epsilon > 0$ from recommending \Act{2} under \State{1} to recommending \Act{1}, and in return move the same amount from recommending \Act{1} under \State{2} to recommending \Act{2}.
  Because the receiver is state-matching, she will still follow the sender's recommendation, and the total probability with which each action is played stays unchanged, so the strategy is still feasible. By Equation~\eqref{eqn:sender-swap}, the sender's utility (weakly) increases. By choosing $\epsilon$ as large as possible, we arrive at the claim or at the following case.
\item If $\SigProbC{1}{2} > 0$ and $\SigProbC{2}{1} = 0$, then $\SigProb{\Act{1}} = \PriorProb \cdot \SigProbC{1}{1} < \PriorProb \leq \UB$. Therefore, it is feasible for the sender to always send the signal \Signal{1} when the state is \State{1}(i.e., decrease \SigProbC{1}{2} to 0 and increase \SigProbC{1}{1} by the same amount). Again, because the receiver is state-matching, she will still follow the sender's recommendation, and the sender is weakly better off because $\Util{S}{\State{1}}{\Act{1}} \geq \Util{S}{\State{1}}{\Act{2}}$.
\end{enumerate}

Because $\Util{S}{\State{2}}{\Act{1}} \geq \Util{S}{\State{2}}{\Act{2}}$ and $\SigProbC{1}{1} = 1$ (as proved above), the sender will also send \Signal{1} with as much probability as possible when the state is \State{2}, subject to not violating the receiver's incentive to play \Act{1} and not exceeding the upper bound \UB (or \UB[']). In other words, the sender maximizes \SigProbC{1}{2} subject to $\Expect[{\STATE \sim \PRIOR, \SIGNAL \sim \SigScheme{\STATE}}]{\Util{R}{\STATE}{\Act{1}}|\Signal{1}} \geq \Expect[{\STATE \sim \PRIOR, \SIGNAL \sim \SigScheme{\STATE}}]{\Util{R}{\STATE}{\Act{2}}|\Signal{1}}$ 
and $\UB\geq \SigProb{\Act{1}}$ (or $\UB[']\geq \SigProb{\Act{1}}$).
Using $\SigProbC{1}{1} = 1$, the incentive constraint is equivalently expressed as $\SigProbC{2}{1} \leq \frac{\PriorProb \cdot \left(\Util{R}{\State{1}}{\Act{1}}-\Util{R}{\State{1}}{\Act{2}}\right)}{(1-\PriorProb) \cdot \left(\Util{R}{\State{2}}{\Act{2}}-\Util{R}{\State{2}}{\Act{1}}\right)}$.
Since this inequality is independent of the bound and \UB['] is more restricted than \UB, the receiver is weakly better off under the constraint \UB than under \UB['].

\item $\Util{S}{\State{1}}{\Act{1}} \leq \Util{S}{\State{1}}{\Act{2}}$
  and $\Util{S}{\State{2}}{\Act{2}} \geq \Util{S}{\State{2}}{\Act{1}}$

This case is symmetric to the previous one. Here, the roles of \Act{1} and \Act{2} (and \State{1} and \State{2}) are reversed, and the important constraint becomes the lower bound \LB (and \LB[']) rather than the upper bound \UB.

\item $\Util{S}{\State{1}}{\Act{1}} \leq \Util{S}{\State{1}}{\Act{2}}$
  and $\Util{S}{\State{2}}{\Act{2}} \leq \Util{S}{\State{2}}{\Act{1}}$

  In this case, the fact that the sender is action-matching together with the assumed inequalities implies that
  \begin{align*}
  \Util{S}{\State{2}}{\Act{2}}
  \stackrel{\text{AM}}{\geq} \Util{S}{\State{1}}{\Act{2}}
  \geq \Util{S}{\State{1}}{\Act{1}}
  \stackrel{\text{AM}}{\geq} \Util{S}{\State{2}}{\Act{1}}
  \geq \Util{S}{\State{2}}{\Act{2}}.
\end{align*}

Thus, the sender's utility is the same, regardless of the state and action.
As a result, the sender is indifferent between all signaling schemes.
In particular, fully revealing the state is an optimal strategy for the sender for any constraint; clearly, this would be best for the receiver.
\end{enumerate}

Thus, for all four cases, the receiver will be no worse off under the more binding constraint.
\end{proof}

\vspace{-6pt}
\subsection{Necessity of Partial Alignment} \label{sec:alignment}
Our main Theorem~\ref{thm:main} assumes that the sender is partially aligned with the receiver (in addition to the state space being binary). One may ask whether the partial alignment is necessary, or whether a more constrained receiver is \emph{always} better off with binary state and action spaces. Here, we show that the assumption is necessary, by giving a $2 \times 2$ example under which the receiver is worse off when more constrained.

The sender's and receiver's utility functions are given in Table~\ref{tb:appendutil}. Here, $0 < \epsilon \ll 1$. The prior distribution over states is $\PriorProb=\frac{1}{4}$.

\begin{table}[ht] 
\vspace{-12pt}
     \caption{Sender's and Receiver's Utility in the example without partial alignment}
   \label{tb:appendutil}
   \vspace{6pt}
  \begin{subtable}[t]{0.48\textwidth}
  \centering
  \begin{tabular}{|c|c|c|}
    \hline
     & $\State{1}$  & $\State{2}$  \\ \hline
    $\Act{1}$ & 2 & 3  \\ \hline
    $\Act{2}$ & 1 & 0   \\ \hline
    \end{tabular}
  \caption{Sender's Utility}
  \end{subtable}
  \begin{subtable}[t]{0.48\textwidth}
  \centering
  \begin{tabular}{|c|c|c|}
    \hline
     & $\State{1}$  & $\State{2}$  \\ \hline
    $\Act{1}$ & 1 & $\epsilon$  \\ \hline
    $\Act{2}$ & 0 & 1   \\ \hline
    \end{tabular}
  \caption{Receiver's Utility} 
  \end{subtable}
  \vspace{-12pt}
\end{table}

The sender prefers action $\Act{1}$ in both states, and the receiver is state-matching. Notice that the sender is not action-matching: when the receiver plays \Act{1}, the sender prefers the state \State{2} over \State{1}.
We write $\Signal{i}$ for the sender's signal suggesting action $\Act{i}$, $i\in \{1,2\}$.

We will compare the receiver's expected utilities in the following two settings:

\begin{enumerate}
    \item There are (effectively) no constraints, i.e., $\UB[\Act{1}]=1,\LB[\Act{1}]=0,\UB[\Act{2}]=1,\LB[\Act{2}]=0$.
    \item The constraint profile binds the sender-preferred action to at most its prior probability, i.e., $\UB[\Act{1}]'=\frac{1}{4},\LB[\Act{1}]'=0,\UB[\Act{2}]'=1,\LB[\Act{2}]'=0$.
\end{enumerate}

The first setting is the classical 
Bayesian persuasion problem: the sender's optimal signaling strategy can be obtained by the concavification approach presented in \cite{kamenica:gentzkow:bayesian}, and is the following:
Send \Signal{1} with $\SigProbC{1}{1}=1$ and $\SigProbC{2}{1}=\third$; send \Signal{2} with $\SigProbC{1}{2}=0$ and $\SigProbC{2}{2}=\frac{2}{3}$.
Given this commitment, the receiver's expected utility is $\frac{1+\epsilon}{2}$ when receiving \Signal{1} (because \State{1} and \State{2} are equally likely to occur), and her expected utility is $1$ when receiving \Signal{2}. Thus, the receiver's overall expected utility is $\frac{3+\epsilon}{4}$.

In the second setting, the sender cannot send the signal \Signal{1} as frequently as in the unconstrained case. When the sender is forced to reduce \Prob{\Signal{1}}, he prefers to reduce the probability \SigProbC{1}{1} instead of \SigProbC{1}{2}. This is because
$\Util{S}{\State{2}}{\Act{1}} - \Util{S}{\State{2}}{\Act{2}}
> \Util{S}{\State{1}}{\Act{1}} - \Util{S}{\State{1}}{\Act{2}}$.
However, reducing \SigProbC{1}{1} solely may cause the signal \Signal{1} to not be persuasive any more, when the posterior belief violates the incentive constraint.
Hence, the sender's optimal signaling strategy requires him to maximize the total probability of \Signal{1}, under the constraint that the receiver is still willing to take action $\Act{1}$ under $\Signal{1}$. Thus, the sender's optimal signaling scheme is the following:
Send \Signal{1} with $\SigProbC{1}{1}=\half$ and $\SigProbC{2}{1}=\frac{1}{6}$; send \Signal{2} with $\SigProbC{1}{2}=\half$ and  $\SigProbC{2}{2}=\frac{5}{6}$.

Against this signaling scheme, the receiver's best response to \Signal{1} is taking action \Act{1}, with an expected utility of $\frac{1+\epsilon}{2}$. Her best response to \Signal{2} is taking action \Act{2}, with an expected utility of $\frac{5}{6}$. Hence, the receiver's expected utility 
is $\frac{6+\epsilon}{8}$ under the constraints \Constr{\LB[']}{\UB[']}.

In summary, the receiver's expected utility of $\frac{3+\epsilon}{4}$ in the first setting is higher than her utility of $\frac{6+\epsilon}{8}$ in the second setting. Thus, we have exhibited an example where a more constrained receiver is worse off than a less constrained one.

\vspace{-6pt}
\section{Failure of the Main Result with larger State Spaces} \label{sec:general}

Unfortunately, contrary to the case of binary state and action spaces, when the state and action spaces are larger, a state-matching receiver and action-matching sender (and feasible constraints) are not enough to ensure that the receiver is always better off when more constrained. Consider the utilities given in Table~\ref{tab:ternary-counterexample}. There are three states in the world, and correspondingly three actions. The prior over the states is uniform.


\begin{table}[ht]
\vspace{-18pt}
  \caption{An example where a constrained receiver is worse off\label{tab:ternary-counterexample}}
  \vspace{6pt}
  \begin{subtable}[t]{0.32\textwidth}
  \centering
  \begin{tabular}{|c|c|c|c|}
    \hline
     & $\State{1}$  & $\State{2}$  & $\State{3}$  \\ \hline
    $\Act{1}$ & 10 & 10 & 0 \\ \hline
    $\Act{2}$ & 0  & 2  & 2 \\ \hline
    $\Act{3}$ & 0  & 0  & 1 \\ \hline
    \end{tabular}
  \caption{Sender's Utility}
  \end{subtable}
  \begin{subtable}[t]{0.32\textwidth}
  \centering
  \begin{tabular}{|c|c|c|c|}
    \hline
     & $\State{1}$  & $\State{2}$  & $\State{3}$  \\ \hline
    $\Act{1}$ & 4 & 2 & 0 \\ \hline
    $\Act{2}$ & 0 & 3 & 1 \\ \hline
    $\Act{3}$ & 0 & 1 & 3 \\ \hline
  \end{tabular}
  \caption{Receiver's Utility}
  \end{subtable}
  \begin{subtable}[t]{0.32\textwidth}
  \centering
   \begin{tabular}{|c|c|c|c|}
 \hline
  & $\State{1}$ & $\State{2}$ & $\State{3}$ \\ \hline
 $\Signal{1}$ & 1 & \half & 0     \\ \hline
 $\Signal{2}$ & 0 & \half & \half \\ \hline
 $\Signal{3}$ & 0 & 0     & \half \\ \hline
 \end{tabular}
\caption{Sender-optimal signaling scheme when $\UB[\Act{1}] = \half$} \label{tab:optimal-constrained-receiver}
 \end{subtable}
  \vspace{-18pt}
\end{table}

Notice that the receiver is state-matching, and the sender is action-matching.

\paragraph{Unconstrained Receiver.}
First, consider an unconstrained receiver.
The sender's optimal signaling scheme \SIGSCHEME is to recommend action \Act{1} whenever the state of the world is \State{1} or \State{2}, and recommend action \Act{3} otherwise.

To verify that the receiver follows the recommendation, one simply compares the utility from the alternative actions: when the sender recommends \Act{1}, following the recommendation gives the receiver expected utility $\half \cdot 4 + \half \cdot 2 = 3$, while \Act{2} would give utility $\half \cdot 0 + \half \cdot 3 = \frac{3}{2}$, and \Act{3} would give $\half \cdot 0 + \half \cdot = \half$. For the recommendation of \Act{3}, the receiver gets to match the state deterministically, so following the recommendation is optimal. Because the signaling scheme is even ex post incentive compatible for the receiver, it is most definitely ex ante incentive compatible.

To see that this signaling scheme is optimal for the sender, first observe that for states \State{1} and \State{2}, the sender obtains the maximum possible utility of 10 over all actions. For state \State{3}, the sender would prefer the receiver to play action \Act{2}. However, the only way to get the sender to play \Act{2} is to mix at least one unit of probability of \State{2} per unit of probability of \State{3}. While this increases the sender's utility for the unit of probability from \State{3} from 1 to 2, it decreases his utility for the unit of probability from \State{2} from 10 (since the receiver played \Act{1}) to 2. Thus, the given signaling scheme is sender-optimal.



Under this signaling scheme, the receiver's expected utility can be calculated as
$\frac{2}{3} \cdot (\half \cdot 4 + \half \cdot 2) + \third \cdot  3 = 3$.

\paragraph{Adding a non-trivial constraint.}
Now, consider a receiver constrained by an upper bound $\UB[\Act{1}]=\half$.
Table~\ref{tab:optimal-constrained-receiver} shows the sender-optimal signaling scheme. Here, the entries show the conditional probability \SigProbC{i}{j} of recommending action \Act{j} (i.e., sending signal \Signal{j}) when the state is \State{i}.


First, notice that action \Act{1} is recommended with probability \half, so the constraint is satisfied.
Second, the receiver will follow the sender's recommendation, as can be checked by comparing her utility from each of the three actions conditioned on any signal. (In the case of receiving \Signal{2}, she is indifferent between \Act{2} and \Act{3} --- recall that we assume tie breaking in favor of the sender.) Again, the given signaling scheme is even ex post incentive compatible, so in particular, it is also ex ante incentive compatible.

To see that the signaling scheme is optimal for the sender, first notice that he induces action \Act{1} (under states \State{1} or \State{2}) with the maximum probability of \half.
Also, notice that using all of the probability from \State{1} to induce \Act{1} is optimal for the sender, because under \State{1}, if any action other than \Act{1} is played, the sender's utility is 0.
Because $\frac{1}{6}$ unit of probability from \State{2} yields a recommendation of \Act{1}, at most $\frac{1}{6}$ can yield a recommendation of \Act{2}, which gives the next-highest utility for the sender.
And because the receiver will choose \Act{2} only when the conditional probability of \State{2} is at least as large as that of \State{3}, action \Act{2} is induced with the maximum possible probability of \third.
Inducing any other actions for any of the states would yield the sender utility 0. Hence, the given signaling scheme is optimal for the sender.


Under this signaling scheme, the receiver's expected utility is
$\half \cdot (\frac{2}{3} \cdot 4 + \third \cdot 2)
  + \third (\half \cdot 3 + \half \cdot 1)
  + \frac{1}{6} \cdot 3
= \frac{17}{6}$.

Thus, the constrained receiver's utility of $\frac{17}{6}$ is lower than the unconstrained receiver's of $3$.

\vspace{-6pt}
\section{Discussion} \label{sec:discuss}
\vspace{-6pt}
We showed that a state-matching receiver, facing an action-matching sender under a binary state space, obtains weakly higher utility when more constrained. We believe that such behavior is in fact observed in the real world: for example, recommenders tend to be more careful in whom they nominate for particularly selective awards or positions.

\vspace{-6pt}
\subsection{Larger State/Action Spaces}
As we discussed in Section~\ref{sec:general}, our results do not carry over to larger state spaces. Indeed, even for state spaces with three states, in which the receiver tries to minimize the distance between the action and the state of the world, there are counter-examples under which a constrained receiver is worse off.

While the result does not hold in full generality with three (or more) states, by imposing additional conditions, a positive result can be recovered:

\begin{proposition} \label{thm:ternary}
  Assume that the state space has size $\SetCard{\StateSpace} = 3$, and that the receiver is state-matching and the sender is action-matching.
  In addition, assume that the following two conditions are satisfied.

\begin{enumerate}
    \item The sender has a monotone\footnote{The result holds symmetrically if the order is reversed.} preference over actions across all states, i.e., $\Util{S}{\State{i}}{\Act{1}} \geq \Util{S}{\State{i}}{\Act{2}} \geq \Util{S}{\State{i}}{\Act{3}}$ for all $i$.
    \item For every state $i$, the receiver is worse off choosing an action $j < i$ that is too low compared to choosing an action $k > i$ that is too high\footnote{Notice that in the case $\SetCard{\StateSpace} = 3$, this constraint only applies to $i=2, j=1, k=3$. We phrase it more generally to set the stage for a further generalization below.}: that is, $\Util{R}{\State{i}}{\Act{j}} \leq \Util{R}{\State{i}}{\Act{k}}$ for all $j < i < k$.
\end{enumerate}

\noindent Then, a more constrained receiver is never worse off than a less constrained one.
\end{proposition}

The additional assumptions on the sender side capture a stronger version of the utility relationship of the interesting cases in the proof of Theorem~\ref{thm:main}. They are motivated in many of our cases: for instance, a letter writer may want to obtain the highest possible honor (or salary) for a student, or a prosecutor may want to maximize the sentence of a defendant.

The additional assumption on the receiver side would capture a cautious department or judge, who would prefer to err on the side of not inviting weak candidates (or giving awards to undeserving candidates), or giving the defendant a sentence that is too low rather than ever giving too high of a sentence.

While Proposition~\ref{thm:ternary} shows that with enough assumptions, a positive result can be recovered, we believe that the assumptions are still rather restrictive, meaning that the proposition is likely of limited interest. The proof involves a long and tedious case distinction, and we therefore do not include it in the paper.



For fully general state spaces (i.e., $n = \SetCard{\StateSpace} \geq 3$), we can currently obtain a positive result only by imposing even more assumptions on the utility functions. In addition to the (generalization of) the assumptions from Proposition~\ref{thm:ternary}, we can make the following assumptions:
(1) Whenever $j < i$, the sender's utility difference between actions $j < j'$ is larger under state \State{i} than under state \State{i'} for $i' > i$. In other words, when the state of the world is smaller, the sender is more sensitive to changes in the receiver's action.
(2) For any fixed state \State{i}, the receiver's utility as a function of $j$ (the action) is increasing and \emph{convex} for $j \leq i$, and decreasing and convex for $j \geq i$.
By adding these two assumptions, we can again obtain a result that a constrained receiver is always weakly better off than an unconstrained one.
While it is possible to construct reasonably natural applications which satisfy these conditions, the conditions are far from covering a broad class of Bayesian persuasion settings.
For this reason, we are not including a proof of this result, instead considering the discussion as a point of departure towards identifying less stringent assumptions that may enable positive results.

Whether there is a broad and natural class of Bayesian persuasion instances with more than two states of the world in which the insight ``A more constrained receiver is better off'' from Theorem~\ref{thm:main} carries over is an interesting direction for future research.

\vspace{-6pt}
\subsection{Finding optimal signaling schemes}
While the main focus of our work is on the receiver's utility when more constrained, our model also raises an interesting computational question, as briefly discussed in Section~\ref{sec:constrained-receiver}.
In particular, we do not know whether there is a polynomial-time algorithm which --- given the sender's and receiver's utility functions as well as the constraints on the receiver --- finds a sender-optimal signaling scheme. Since probability constraints on receivers (quotas) are quite natural in many signaling settings, this constitutes an interesting direction for future work.

The main difficulty in applying standard techniques is that the constraints may force the receiver to play an ex post suboptimal action. The standard LP for the sender's optimization problem \cite{dughmi:xu:algorithmic:2019}
maximizes the sender's expected utility subject to the constraint that the receiver is incentivized to play the sender's recommended action.
To appreciate the difference, consider a setting in which the state of the world is uniform over $\SET{\State{1}, \State{2}}$, and the sender and receiver both obtain utility 1 if the receiver plays action \Act{1}, and 0 otherwise.
Without any constraints, the sender need not send any signal, and the receiver would simply play action \Act{1}.
But if the receiver is constrained to playing action \Act{1} with probability exactly $\half$, then she must randomize, including the (always suboptimal) action \Act{2} with probability \half.
By Proposition~\ref{prop:correspondingpure}, the randomization can be pushed to the sender instead, but when the sender recommend action \Act{2}, it will be ex post suboptimal for the receiver to follow the recommendation.
Indeed, an LP requiring deterministic ex post obedience from the sender would become infeasible for this setting.
Whether the sender's optimization problem can still be cast as a different LP, or solved using other techniques, is an interesting direction.

We remark here that the preceding example does not have a state-matching receiver. If the receiver is state-matching and the constraints are feasible, then full revelation of the state is ex post incentive compatible for the receiver. This implies that the linear program for optimizing the sender's utility over ex post incentive compatible signaling schemes has a feasible solution. However, since the LP is more restricted, it is not at all clear that its optimum solution maximizes the sender's utility when the recommendation does not have to be ex post incentive compatible.

\vspace{-12pt}
\subsection{Receiver's strategic behaviors on constraint enforcement}

  We assumed throughout the paper that the receiver's constraints are common knowledge, and that enforcing the constraints is indeed required of the receiver (or in her best interest).
  Aside from the interview example provided in Section~\ref{sec:intro}, such constraints are encountered in real-world scenarios such as a patient's dietary restrictions, the salary cap for a sports team, or the capacity limit of an event or facility.

  Given that we showed constraints to be beneficial for the receiver, one may suspect that a receiver could strategically misrepresent how harsh her constraints are, or --- along the same lines --- claim to be constrained, but not enforce the claimed constraints. This would allow the receiver to obtain more information from a sender. In other words, when constraints are not common knowledge, they become private information of the receiver, which could be strategically manipulated; for instance, in the interview example, Alice could indicate a constraint just to force Bob's hand.

    Naturally, allowing strategic manipulation in the model will significantly complicate the problem, either making it a dynamic information design problem \cite{farhadi:teneketzis:dynamicID:2021} with multiple senders \cite{ambrus:takahashi:multicheaptalk:2008} or a mechanism design problem with incorporated information design modules \cite{roesler2014mechanism}.
    Analyzing a model with private receiver constraints thus constitutes an interesting directions for future work.

\vspace{-8pt}
 \bibliographystyle{splncs04}
\bibliography{local-bibliography,davids-bibliography/names,davids-bibliography/conferences,davids-bibliography/bibliography,davids-bibliography/publications}

\end{document}